\documentclass[sigconf]{acmart}

\usepackage{subfig, physics}

\AtBeginDocument{%
 \providecommand\BibTeX{{%
 \normalfont B\kern-0.5em{\scshape i\kern-0.25em b}\kern-0.8em\TeX}}}

\acmConference[PASC '22]{PASC '21: ACM Platform for Advanced Scientific Computing Conference}{June 27--29, 2022}{Basel, CH}
\acmBooktitle{PASC '22: ACM Platform for Advanced Scientific Computing Conference,
 June 27--29, 2022, Basel, CH}
\begin{document}

\title{Communication Bounds for Convolutional Neural Networks}

\author{Anthony Chen}
\authornote{Authors listed alphabetically by last name}
%\email{cygnari@umich.edu}
\orcid{0000-0003-0760-8329}
%\email{webmaster@marysville-ohio.com}
\affiliation{%
 \institution{University of Michigan, Ann Arbor}
% \streetaddress{530 Church St}
 \city{Ann Arbor}
% \state{Michigan}
 \country{USA}
% \postcode{48109}
}

\author{James Demmel}
\affiliation{%
 \institution{University of California, Berkeley}
% \streetaddress{564 Soda Hall}
 \city{Berkeley}
 \country{USA}}
%\email{demmel@berkeley.edu}

\author{Grace Dinh}
\affiliation{%
 \institution{University of California, Berkeley}
 \city{Berkeley}
 \country{USA}
}

\author{Mason Haberle}
\affiliation{%
 \institution{New York University}
 \city{New York}
 \country{USA}}

\author{Olga Holtz}
\affiliation{%
 \institution{University of California, Berkeley}
% \streetaddress{30 Shuangqing Rd}
 \city{Berkeley}
% \state{Beijing Shi}
 \country{USA}}

\renewcommand{\shortauthors}{Chen et al.}

\begin{abstract}
Convolutional neural networks (CNNs) are important in a wide variety of machine learning tasks and applications, so optimizing their performance is essential. Moving words of data between levels of a memory hierarchy or between processors on a network is much more expensive than the cost of arithmetic, so minimizing communication is critical to optimizing performance. In this paper, we present new lower bounds on data movement for mixed precision convolutions in both single-processor and parallel distributed memory models, as well as algorithms that outperform current implementations such as Im2Col. We obtain performance figures using GEMMINI, a machine learning accelerator, where our tiling provides improvements between 13\% and 150\% over a vendor supplied algorithm. 
% add mixed precision
% update attainability, performance results 

\end{abstract}



%\ccsdesc[500]{Computer systems organization~Embedded systems}


\maketitle

\section{Introduction} \label{intro}
Convolutional neural networks (CNNs) are important in many machine learning applications and their computational intensity makes their computation a major bottleneck, requiring efficient implementations on modern architectures. To do so, it is important to recognize that most of the time and energy spent during the execution of a CNN often goes towards communication, the movement of data between different levels of the memory hierarchy (RAM to L3 cache) or between processors operating in parallel. The cost of moving one word of data is frequently orders of magnitude larger than the cost of performing one arithmetic operation both in terms of time and power consumption. This disparity is only increasing as time passes \cite{NAP12980}. Minimizing communication has driven optimization efforts in numerical linear algebra, giving rise to the highly tuned implementations seen in BLAS and LAPACK which attain a high fraction of a machine's maximum possible performance. 

%More recently, \cite{dd18} has begun the process of extending this kind of optimization to the domain of CNNs. We continue this line of inquiry. 

In this paper, we consider the problem of computing a single convolution layer of a CNN, which can be written as seven nested loops. Our model is described in Section~\ref{cnn}. In our theoretical work, we consider various ways of organizing this computation, and we ask which order minimizes the amount of communication between main memory and cache in the single processor case, or between the network of processors in the parallel case. We describe our CNN and memory model in detail in~\ref{cnn}. %This sentence is a bit awkward.
%Lower bounds only claim the amount of communication one processor does. This is all that's relevant for time. We may want to add a total count too.

Our first contribution is to provide new communication lower bounds in both single processor and parallel architectures including precise constants and allowing for mixed precision data. These bounds are presented in Section \ref{results} with proofs in Sections \ref{serialbounds} and \ref{parallelbounds}. Our second contribution is to provide algorithms which meet these bounds in the parallel case for large parameters, and which approach them more closely than previously attainable in the single processor case. These results are described in Sections \ref{singleattainability} and \ref{parallelattainability}.

%In applications, oftentimes the data comes at mixed levels of precision. Our third contribution is to generalize our lower bounds and attainability analysis to include the possibility of varying the relative precisions of the incoming image data, the filters, and the output of convolution. The generalized results are presented in Section~\ref{results} and their proofs in Appendix~\ref{proofs}.

The rest of the paper is organized as follows. Section~\ref{preliminaries} describes our CNN model before we briefly discuss our results, and introduces the main theoretical tool for computing lower bounds, the H\"{o}lder-Brascamp-Lieb inequalities. Sections~\ref{single} and \ref{parallel} present lower bounds for the single processor and parallel cases respectively with discussions of attainability. In Section~\ref{performance}, we discuss performance results obtained using GEMMINI, a machine learning accelerator. 

\section{Preliminaries} \label{preliminaries}

\subsection{CNN and Memory Model} \label{cnn}

We consider the following loop nest for directly computing a convolution layer of a CNN. Since it has 7 nested loops surrounding a simple update instruction, we call it 7NL CNN:
\begin{align}
 & {\text{for}}\{i_1,i_2,i_3,i_4,i_5,i_6,i_7\}=0:\{N,c_I,c_O,w_O,h_O,w_F,h_F\}-1\nonumber \\
 & \ \ \text{Output}(i_1,i_3,i_4,i_5)\,+=\text{Input}(i_1,i_2,\sigma_{w}i_4 + i_6, \sigma_{h}i_5 + i_7)\times\nonumber \\& \ \ \quad\quad\quad\quad\quad\quad\quad\quad\quad\quad \text{Filter}(i_2,i_3,i_6,i_7)\label{eqn_CNN}
\end{align}
where the Input has dimensions $N\times c_I\times(w_F + \sigma_{w}w_O)\times(h_F + \sigma_{h}h_O)$, the Output has dimensions $N\times c_O\times w_O\times h_O$, and the Filter has dimensions $c_I\times c_O\times w_F\times h_F$. $N$ is the number of images, $c_I$ is the number of channels of the input image, $c_O$ is the number of channels of the output image, $w_O$ and $h_O$ are the width and height of the output image, $w_F$ and $h_F$ are the width and height of one convolution filter, and $\sigma_{w}$ and $\sigma_{h}$ are the stride sizes in the horizontal and vertical dimensions respectively. We assume that the filter sizes are smaller than the input image sizes, and in practice, they are usually much smaller. This gives us the assumptions $w_F\leq\sigma_ww_O$ and $h_F\leq\sigma_hh_F$. We also assume that $\sigma_w\leq w_F$ and $\sigma_h\leq h_F$ so that all elements of the image are used. Then, the input has size $|I| = Nc_I(\sigma_ww_O+w_F)(\sigma_hh_O+h_F)$, the output has size $|O| = Nc_Ow_Oh_O$, and the filter has size $|F| = c_Ic_Ow_Fh_F$. The precisions of the input, output, and filter are $p_I$, $p_O$, and $p_F$ respectively.  These are in units of words (32 bits).  We define the sum of the precisions as $p_T=p_I+p_O+p_F$. 

Note that each iteration of the loop nest requires access to a single element of both the Input and Filter arrays, and must make a single update to the Output array. The order in which these updates are made does not impact the result, so they may be reorganized as desired to optimize for data movement. We consider a \emph{computation} of the 7NL CNN algorithm to be an execution of all $G := Nc_Ic_Ow_Oh_Ow_Fh_F$ updates, performed in any order.

We consider computations of 7NL CNN within two different architectures. The first is a single processor architecture with a 2-layer memory model: a cache which may hold $M$ words of data and may be accessed instantaneously by the processor, and a main memory of arbitrary size. The movement of one word of data from the main memory to the cache, or back, is counted as a single unit of communication. Data begins in the main memory, and the output of the computation must reside in the main memory before the computation is complete.

The second architecture is a parallel processor architecture with a distributed memory model. There are $P$ processors, each with their own local instantaneously accessible memory of size $M$ words. Any one processor may communicate with any other processor, and each word sent is counted as a single unit of communication. Data may begin in any processor, and may end in any processor by the completion of the computation.
%\subsection{Previous Bounds} \label{prev}
%
%We now review the previously known asymptotic bounds for computations of 7NL CNN. For the single processor case, \cite{dd18} derives a communication bound of
%\[ X=\Omega\Bigg(\max\Bigg\{|I|, |F|, |O|,\frac{G}{M}, \frac{G(\sigma_w \sigma_h)^{1/2}}{(w_Fh_FM)^{1/2}}\Bigg\}\Bigg) \]
%where $X$ is the number of words that must be communicated. Additionally, a tiling method analogous to that used for matrix multiplication is provided and is shown to be asymptotically optimal. 
%
%In the parallel case, \cite{dd18} uses a different architecture. They assume that the processors communicate with a single shared memory, and do not communicate with one another. Within this architecture, they derive the communication bound
%
%\[X=\Omega\Bigg(\max\Bigg\{|O|, a_{c_O}|I|, a_Na_{w_O}a_{h_O} |F|,\frac{G}{M}, \frac{G(\sigma_w \sigma_h)^{1/2}}{(w_Fh_FM)^{1/2}}\Bigg\}\Bigg)\]
%where $X$ is the total number of communications done, summed over all processors. In particular, some processor must perform $X / P$ communications. In this bound, the new parameters $a_N$, $a_{c_O}$, $a_{w_O}$, and $a_{h_O}$ come from the way the operations are among processors. For example, each processor handles $N'=\frac{N}{a_N}$ as its batch size). This leads to the constraint $a_Na_{c_O}a_{w_O}a_{h_O}\leq P$, and the parameters can be chosen to maximize the lower bound. Even though we use a distributed memory architecture instead of a shared memory model, this lower bound is still a helpful reference for our results in the parallel case.
\subsection{Our Results} \label{results}

Before we discuss the mathematical tools involved in our analysis, we first briefly summarize our results. We find results corresponding both to the single and parallel architectures with mixed precision data. Proofs are found in Sections \ref{serialbounds} and \ref{parallelbounds}.

We find the following lower bound for the number of words communicated in a single-processor memory model with fast memory (cache) and slow memory, where the Input array takes entries which are $p_I$ words in length, and similarly $p_F$ for the Filter and $p_O$ for the Output.
\begin{theorem} \label{thm:singleBoundPrecision}
If $X$ is the number of words communicated by \emph{7NL CNN} within a single-processor memory model with $M$ words of fast memory, where \emph{Input}, \emph{Filter}, and \emph{Output} are non-overlapping arrays and $G := Nc_Ic_Ow_Oh_Ow_Fh_F$ is the total number of updates performed during the computation, we have

\begin{multline*}
 X \geq \max\big\{p_I|I| + p_F|F| + p_O|O|, C_pGM^{-1} - M, \\ {2(p_Ip_Fp_O)^{1/2}(\sigma_w\sigma_h)^{1/2}G }{(w_Fh_FM)^{-1/2}} - 2M\big\}
\end{multline*}

where the value of $C_p = C_p(p_I, p_F, p_O)$ depends on the precisions satisfying a triangle condition:
\[
 C_p = \begin{cases} \frac{1}{4}p_T^2 & p_j \leq p_k + p_\ell \quad \textrm{for all distinct } j,k,\ell\\
% \quad \textrm{for all distinct } j, k, \ell \in \{I, F, O\} \\
 p_j(p_k + p_\ell) & p_j > p_k + p_\ell \quad \textrm{for some distinct } j, k, \ell 
% \in \{I, F, O\}
 \end{cases}
\]
\end{theorem}

In the standard case when each matrix has precision $1$, $C_p = 9/4$. The first bound corresponds to accessing each memory location at least once. The second bound dominates when individual $w_F \times h_F$ filters are large relative to the memory size $M$, and the third bound dominates when filters are small relative to $M$. In all practical cases, the precisions satisfy the triangle condition, so the first expression for $C_p$ is more relevant.

%This needs to be adjusted to remove appendix 1.
We additionally discuss a blocking technique for evaluating 7NL CNN which comes close to meeting this bound. Plots depicting this attainability are found in Figure \ref{fig:serial1}. %and Appendix \ref{graphs}. 

We also derive the following lower bounds for the number of words communicated by some processor in a parallel-processor distributed memory model. Again, these bounds accept mixed precision data. The first two bounds are similar to our bounds in Theorem \ref{thm:singleBoundPrecision}, decaying with the memory size.

\begin{theorem}
\label{thm:parallelBoundMDPrecision}
If $X$ is the number of words communicated by \emph{7NL CNN} within a parallel processor memory model with $P$ processors each with $M$ words of memory, where \emph{Input}, \emph{Filter}, and \emph{Output} are nonoverlapping arrays with precisions $p_I, p_F, p_O$ respectively and $G := Nc_Ic_Ow_Oh_Ow_Fh_F$ is the total number of updates performed during the computation, we have
\[
 X \geq \max\left\{\frac{C_pG}{PM} - M, \frac{2(p_Ip_Fp_O)^{1/2}(\sigma_w\sigma_h)^{1/2}G }{P(w_Fh_FM)^{1/2}} - 2M\right\}
\]
with $C_p$ as in Theorem \ref{thm:singleBoundPrecision}.
\end{theorem}

The above bounds are only nontrivial when $M$ and $P$ are small relative to $G$. If there are many processors, or each processor has access to more memory, then we require new lower bounds corresponding in spirit to the 2.5D matrix multiplication data replication algorithms presented in \cite{de13}. These two bounds are memory independent and require a load-balancing assumption.

\begin{theorem}
\label{thm:parallelBoundMIPrecision}
Consider an execution of \emph{7NL CNN} within a parallel processor memory model with the setup of Theorem \ref{thm:parallelBoundMDPrecision}. Suppose further that each array is initially load balanced and that and $A_p := \max\{p_I|I|, p_F|F|, p_O|O|\}$ is the memory size of the largest array. Some processor must communicate $X$ words, where
\[
 X \geq (p_Ip_Fp_O)^{1/3}\max\left\{\frac{G^{1/2}}{P^{1/2}},
 \frac{(G\sigma_w\sigma_h)^{2/3}}{(Pw_Fh_F)^{2/3}} \right\} - \frac{A_p}{P}.
\]
\end{theorem}

%Are we presenting careful proofs?
We introduce the machinery used to obtain these lower bounds, and then present careful proofs.

\subsection{The H\"older-Brascamp-Lieb Inequality} \label{hbl}

One of the key mathematical tools for proving theoretical lower bounds for the communication cost of nearly any numerical linear algebra computation is the collection of H\"older-Brascamp-Lieb inequalities. These inequalities bound the sizes of sets in terms of the sizes of their linear projections.  Their simplest form, the Loomis-Whitney inequality, is already useful for bounding the communication cost of matrix multiplication \cite{bcdhks14}.  When $A \subset \mathbb{Z}^3$ is a finite set of integer lattice points and $A_x$, $A_y$, and $A_z$ its 2D projections along each axis, the Loomis-Whitney inequality states
\[ |A| \leq \sqrt{|A_x| |A_y| |A_z|}. \]
In words, the size of $A$ is bounded in terms of the sizes of its shadows.
\begin{figure}[ht]
  \label{fig:loomiswhitney}
  \centering
  \includegraphics[width=5cm]{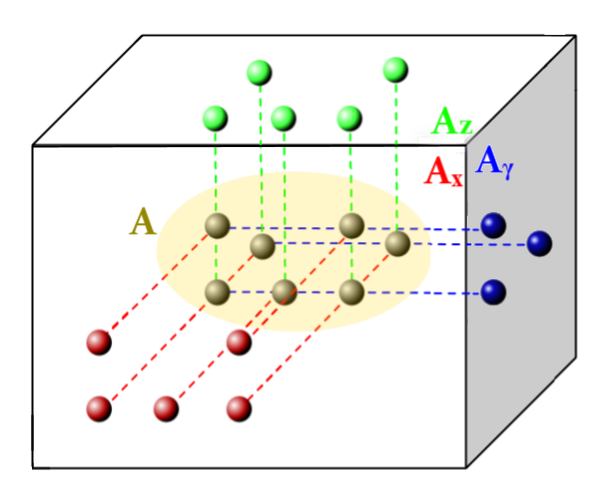}
  \caption{Loomis-Whitney inequality $|A| \leq \sqrt{|A_x| |A_y| |A_z|}$.}
  \end{figure}
%They generalize geometric bounds such as the Loomis-Whitney inequality as well as functional-analytic bounds such as H\"older's inequality. 

To perform 3-nested-loop matrix multiplication, the $(i, j, k)$th operation in computing $C = AB$ is setting $C[i, k] = A[i, j]B[j, k]$.  But these locations are the projections.  So one must access \\
$|A_x| + |A_y| + |A_z|$ memory locations to perform all $|A|$ operations.  The Loomis-Whitney inequality tells us how many operations we can perform when accessing a limited number of memory locations.

This idea can be generalized to loop nests of arbitrary size whose array accesses are affine functions of the indices. The key tool is a generalization of the Loomis-Whitney inequalities, a discrete formulation of the H\"older-Brascamp-Lieb inequalities:

\begin{theorem}[Discrete HBL \cite{cdksy13}]
\label{thm:HBL}
Let $d$ and $d_j$ be nonnegative integers, and for $j = 1, \dots, m$ let $\phi_j : \mathbb{Z}^d \to \mathbb{Z}^{d_j}$ be group homomorphisms. If $s_j \in [0, 1]$ for $j = 1,\dots, m$ satisfy the following collection of linear constraints:
\[ \operatorname{rank}(H) \leq \sum_{j = 1}^m s_j\operatorname{rank}(\phi_j(H)) \textrm{ for each subgroup } H \leq \mathbb{Z}^d\]
then we find the bound
\[ |V| \leq \prod_{j = 1}^m |\phi_j(V)|^{s_j} \textrm{ for each finite set } V \subseteq \mathbb{Z}^d.\]
\end{theorem}

This result is stated and proved in detail in Section 3 of \cite{cdksy13}. We call such a tuple $(s_1, \dots, s_m)$ \emph{HBL exponents} for the \emph{HBL datum} $(\mathbb{Z}^d, (\mathbb{Z}^{d_j}), (\phi_j))$ and attempt to minimize the sum of the $s_j$ subject to the linear constraints.

These inequalities provide a powerful tool for understanding data movement. For our purposes, they allow us to bound the number of updates that an execution of the 7NL CNN algorithm is able to complete if it is only allowed a certain number of accesses to each array. More specifically, suppose we run 7NL CNN within a single-processor one level memory model with cache size $M$. We allow the updates to be executed in an arbitrary order. Consider a continuous segment of updates during the execution which makes exactly $T$ communications with main memory, whether loads from Input or Filter, or stores to Output. At the start of the segment, we have access to no more than $M$ elements from any of the three arrays, and we may load no more than $T$ elements during the segment. So the number of elements accessible from each array is at most $M + T$.

Now define three \emph{array access homomorphisms} which map any tuple of loop indices $(i_j)$ to each of the tuples of indices accessed in each array. For example, $\phi_O : \mathbb{Z}^7 \to \mathbb{Z}^4$ defined by $\phi_O(i_1, \dots, i_7) = (i_1, i_3, i_4, i_5)$ is the array access homomorphism for the Output array as seen in the model presented in Section \ref{cnn}. Then any valid tuple of HBL exponents $(s_j)$ satisfying the constraints of Theorem \ref{thm:HBL} provides a bound on the size of $V$, the set of updates the execution may compute during this segment, in terms of the sizes of the $\phi_j(V)$'s, the number of elements which need to be accessed from each array:
\[ |V| \leq \prod_{j=1}^m |\phi_j(V)|^{s_j} \leq \prod_{j=1}^m (M + T)^{s_j} = (M + T)^{\sum_j s_j}\]
Let $s = \sum_j s_j$. Now, if the entire execution must make $G$ updates, then splitting the entire execution up into $L$ segments each with about $T$ communications, each segment may do no more than $(M + T)^s$ updates and so there must be at least $G/(M + T)^s - 1$ segments. Each has $T$ communications, so the total number of communication must be at least $GT/(M + T)^s - T$. Choosing $T = M$ for now, we find the number of communications to be $\Omega(G/M^{s-1})$. If one takes the time to compute the optimal $(s_j)$ for various choices of array access homomorphisms, this exact same proof sketch results in the well-known asymptotic communication lower bounds one finds for many common linear algebra algorithms such as those tabulated in \cite{bcdhks14}. We now discuss how to make practical use of Theorem \ref{thm:HBL}. Note that because each rank is an integer between $0$ and $d$, the number of constraints that we need to check for a given tuple $(s_1, \dots, s_j)$ is finite. We can further reduce the workload by showing that the only subgroups which need to be checked are those in the subgroup lattice generated by the kernels of the $\phi_j$. The \emph{lattice} generated by a family of subgroups is the smallest collection of subgroups which contains this family and which is closed under intersection and sum of subgroups. We denote our lattice of interest by $\operatorname{Lattice}(\ker \phi_j)$. This is formalized in the following proposition:
%, which is similar in flavor to the discussion of subgroup lattices in \cite{dd18}:

%This proposition is kind of silly. Might be able to be removed.
\begin{proposition}
\label{prop:kerReduction}
Let $\phi_j$ be homomorphisms and $s_j$ exponents as in Theorem \ref{thm:HBL}, and consider $\operatorname{Lattice}(\ker \phi_j)$, the lattice generated by the subgroups $\ker \phi_j$. If
\[ \operatorname{rank}(H) \leq \sum_{j = 1}^m s_j\operatorname{rank}(\phi_j(H)) \textrm{ for each } H \in \operatorname{Lattice}(\ker \phi_j)\]
then
\[ \operatorname{rank}(H) \leq \sum_{j = 1}^m s_j\operatorname{rank}(\phi_j(H)) \textrm{ for each subgroup } H \leq \mathbb{Z}^d\]
and the conclusion of HBL follows.
\end{proposition}

\begin{proof}
Let $\Phi_j : \mathbb{Q}^d \to \mathbb{Q}^{d_j}$ extend $\phi_j$ to a $\mathbb{Q}$-linear map. In Section 2.2 of \cite{cdksy15}, it is proved that the polytope of $(s_j)$ satisfying the linear constraints in Theorem \ref{thm:HBL} is exactly equal to the polytope of $(s_j)$ satisfying, for each subspace $V \leq \mathbb{Q}^d$,
\[ \dim V \leq \sum_{j=1}^m s_j \dim \Phi_j(V). \]
In \cite{val10}, Theorem 8 states that it suffices to check only these inequalities from subspaces in $\operatorname{Lattice}(\ker \Phi_j)$. To check these, it suffices to check the original inequalities from Theorem \ref{thm:HBL} on the subgroups in $\operatorname{Lattice}(\ker \phi_j)$ and then to take $\mathbb{Q}$-linear spans of these subgroups. This completes the proof.
\end{proof}
We can reduce the number of subgroups to check even further. We do this by splitting $\operatorname{Lattice}(\ker \phi_j)$ into several independent sublattices. We call two collections of subgroups $\{A_i\}$ and $\{B_k\}$ \emph{independent} if $\sum_i A_i \cap \sum_k B_k = \{0\}$. If $A = \{A_i\}$ and $B = \{B_k\}$ are independent collections, it quickly follows that
\begin{align*}
 \operatorname{Lattice}(A \cup B) &= \operatorname{Lattice}(A) + \operatorname{Lattice}(B) \\
 \operatorname{rank}(A_i + B_k) &= \operatorname{rank}(A_i) + \operatorname{rank}(B_k) \\
 \operatorname{rank}(\phi_j(A_i + B_k)) &= \operatorname{rank}(\phi_j(A_i)) + \operatorname{rank}(\phi_j(B_k)).
\end{align*}
Then if $\operatorname{Lattice}(\ker\phi_j) = \sum_k \operatorname{Lattice}(A^k)$ where the $A^k$'s are pairwise independent collections of subgroups, it suffices to check the constraints in Theorem \ref{thm:HBL} on the subgroups in each $\operatorname{Lattice}(A^k)$. As we will see, these reductions are capable of reducing the number of linear constraints we need to check from hundreds to only a few. We will leverage this to compute optimal exponent tuples $(s_j)$ for the array access homomorphisms in 7NL CNN.

%Serial instead of single processor?
\section{Single Processor Communication Bounds} \label{single}

\subsection{Derivation of New Bounds}
\label{serialbounds}

%Will we recompute it?
In this section, we prove Theorem \ref{thm:singleBoundPrecision}. Performing an analysis with the HBL inequalities, we derive precise lower bounds, taking care to optimize constants.
%Following the HBL analysis in \cite{dd18} (which we recompute here) we turn the previously known asymptotically optimal bounds into precise lower bounds, taking care to optimize constants.
For example, in the standard precision case $p_I = p_F = p_O = 1$, the bound becomes
\[
 X \geq \max\left\{|I| + |F| + |O| , \frac{9G}{4M} - M, \frac{2G(\sigma_w\sigma_h)^{1/2}}{(w_Fh_FM)^{1/2}} - 2M\right\}
\]
The first bound doesn't depend on the memory size. The second bound exhibits $\Omega(1/M)$ decay, while the third bound exhibits $\Omega(1/M^{1/2})$ decay. However, it is important to note that the third bound only eclipses the second bound when $w_F h_F < \frac{64 M \sigma_w \sigma_h}{81}$, i.e. when the filters are small relative to the memory size.

%Might not provide optimization process anymore.
We prove each of the three bounds below. For the second and third bound, we will make use of the HBL theory discussed in Section \ref{hbl}. First, we have a trivial memory-independent bound:

\begin{lemma}
\label{lem:singleTrivPrecision}
With the setup in Theorem \ref{thm:singleBoundPrecision}, the number of words communicated $X$ satisfies
\[
 X \geq p_I|I| + p_F|F| + p_O|O|
\]
\end{lemma}

\begin{proof}
Every entry of Input and Filter must be accessed at least once, and every entry in Output must be filled by the computation. All three arrays reside in slow memory, so at minimum $p_j$ words must be communicated for every entry in the $j$th array, for $j \in \{I, F, O\}$. So the number of words communicated $X$ satisfies
\[
X\geq p_I\abs{I}+p_F\abs{F}+p_O\abs{O}\qedhere
\]
\end{proof}

Before proving the second bound, we perform an HBL analysis on the array-access homomorphisms corresponding to 7NL CNN. First, we define the homomorphisms $\phi_I, \phi_F, \phi_O : \mathbb{Z}^7 \to \mathbb{Z}^4$:
\begin{align*}
\phi_I(i_1, i_2, i_3, i_4, i_5, i_6, i_7) &= (i_1, i_2, i_6 + \sigma_w i_4, i_7 + \sigma_h i_5) & \\
\phi_F(i_1, i_2, i_3, i_4, i_5, i_6, i_7) &= (i_2, i_3, i_6, i_7) & \\
\phi_O(i_1, i_2, i_3, i_4, i_5, i_6, i_7) &= (i_1, i_3, i_4, i_5) &
\end{align*}
Note that in Section \ref{cnn} the iteration of 7NL CNN corresponding to the loop indices $(i_j)$ uses the data at Input$(\phi_I(i_j))$ and Filter$(\phi_F(i_j))$ to update the data at Output$(\phi_O(i_j))$. The discussion in Section \ref{hbl} suggests that we analyze the lattice generated by the kernels of these homomorphisms. Using the indices $i_j$ as free variables in $\mathbb{Z}$, we can write the kernels as follows:
\begin{align*}
\ker \phi_I &= (0, 0, i_3, i_4, i_5, -\sigma_w i_4, -\sigma_h i_5) & \\
\ker \phi_F &= (i_1, 0, 0, i_4, i_5, 0, 0) & \\
\ker \phi_O &= (0, i_2, 0, 0, 0, i_6, i_7) &
\end{align*}
%Following the work of \cite{dd18}, 
We can identify the following independent families of indices: $\{i_1\}$, $\{i_2\}$, $\{i_3\}$, $\{i_4, i_6\}$, and $\{i_5, i_7\}$. We call these independent because they give rise to the following pairwise independent collections of subgroups which generate the kernels we want, and hence the lattice we want:
\begin{align*}
C_1 &= \{(i_1, 0, 0, 0, 0, 0, 0)\} = \{C_{1, 1}\} \\
C_2 &= \{(0, i_2, 0, 0, 0, 0, 0)\} = \{C_{2, 1}\} \\
C_3 &= \{(0, 0, i_3, 0, 0, 0, 0)\} = \{C_{3, 1}\} \\
C_4 &= \{(0, 0, 0, i_4, 0, 0, 0), (0, 0, 0, 0, 0, i_6, 0), (0, 0, 0, i_4, 0, -\sigma_w i_4, 0)\} \\
 &= \{C_{4, 1}, C_{4, 2}, C_{4, 3}\} \\
C_5 &= \{(0, 0, 0, 0, i_5, 0, 0), (0, 0, 0, 0, 0, 0, i_7), (0, 0, 0, 0, i_5, 0, -\sigma_h i_5)\} \\
&= \{C_{5, 1}, C_{5,2}, C_{5,3}\}
\end{align*}
These subgroups give the following breakdown of the kernels:
\begin{align*}
\ker \phi_I &= C_{3,1} + C_{4,3} + C_{5,3} & \\
\ker \phi_F &= C_{1,1} + C_{4,1} + C_{5,1} & \\
\ker \phi_O &= C_{2,1} + C_{4,2} + C_{5,2} & 
\end{align*}
In order to apply Theorem \ref{thm:HBL}, the discussion in Section \ref{hbl} concludes that it suffices to check the constraints only on subgroups in the five lattices, $\operatorname{Lattice}(C_j)$. Fortunately, $\operatorname{Lattice}(C_j) = C_j$ for $j = 1,2,3$. For $C_4$ and $C_5$ we have:
\begin{align*}
\operatorname{Lattice}(C_4) &= C_4 \cup \{(0, 0, 0, i_4, 0, i_6, 0)\} = C_4 \cup \{C_{4, 4}\} \\
\operatorname{Lattice}(C_5) &= C_5 \cup \{(0, 0, 0, 0, i_5, 0, i_7)\} = C_5 \cup \{C_{5, 4}\}
\end{align*}
Now suppose $s_I, s_F, s_O \in [0, 1]$. Then to apply Theorem \ref{thm:HBL}, we need to satisfy the following inequality for each $H$ in some $\operatorname{Lattice}(C_j)$:
\[ \operatorname{rank}(H) \leq s_I \operatorname{rank}(\phi_I(H)) + s_F \operatorname{rank}(\phi_F(H)) + s_O \operatorname{rank}(\phi_O(H)). \]
We enumerate these inequalities in the table below:
\begin{center}
\begin{tabular}{|c|c|c|c|c|c|}
\hline
$H$ & $\mathrm{rk}$ & $\mathrm{rk}\circ \phi_I$ & $\mathrm{rk}\circ \phi_F$ & $\mathrm{rk}\circ \phi_O$ & Constraint \\
\hline
$C_{1,1}$ & 1 & 1 & 0 & 1 & $1 \leq s_I + s_O$ \\
$C_{2,1}$ & 1 & 1 & 1 & 0 & $1 \leq s_I + s_F$ \\
$C_{3,1}$ & 1 & 0 & 1 & 1 & $1 \leq s_F + s_O$ \\
$C_{4,1}$ & 1 & 1 & 0 & 1 & $1 \leq s_I + s_O$ \\
$C_{4,2}$ & 1 & 1 & 1 & 0 & $1 \leq s_I + s_F$ \\
$C_{4,3}$ & 1 & 0 & 1 & 1 & $1 \leq s_F + s_O$ \\
$C_{4,4}$ & 2 & 1 & 1 & 1 & $2 \leq s_I + s_F + s_O$ \\
$C_{5,1}$ & 1 & 1 & 0 & 1 & $1 \leq s_I + s_O$ \\
$C_{5,2}$ & 1 & 1 & 1 & 0 & $1 \leq s_I + s_F$ \\
$C_{5,3}$ & 1 & 0 & 1 & 1 & $1 \leq s_F + s_O$ \\
$C_{5,4}$ & 2 & 1 & 1 & 1 & $2 \leq s_I + s_F + s_O$ \\
\hline
\end{tabular}
\end{center}
Removing repeated inequalities, Theorem \ref{thm:HBL} states that as long as $1 \leq s_I + s_F$, $1 \leq s_I + s_O$, $1 \leq s_F + s_O$, and $2 \leq s_I + s_F + s_O$, then for any finite set $V \subseteq \mathbb{Z}^7$ we have:
\[|V| \leq |\phi_I(V)|^{s_I}|\phi_F(V)|^{s_F}|\phi_O(V)|^{s_O}. \]

We can now begin the proof of the second bound. First, we handle the case when the triangle condition is met.

\begin{lemma}
\label{lem:singleLFPrecisionTriangle}
With the setup in Theorem \ref{thm:singleBoundPrecision}, as long as $p_I \leq p_F + p_O$, $p_F \leq p_I + p_O$, and $p_O \leq p_I + p_F$, the number of words communicated $X$ satisfies
\[
   X \geq \frac{(p_I + p_F + p_O)^2G}{4M} - M
\]
\end{lemma}

\begin{proof}

We split the execution of the 7NL CNN computation into $L$ segments of contiguous updates. In each segment, we allow exactly $T$ words to be loaded/stored, except for possibly the last segment which may have $\leq T$ words.

We fix our attention to a single segment. Let $V$ be the set of indices of updates computed during the current segment. $V$ contains tuples $(i_j) \in \mathbb{Z}^7$. Then $\phi_I(V)$ is the set of indices of Input whose data must be accessed during the segment, $\phi_F(V)$ the set of indices of Filter, and $\phi_O(V)$ the set of indices of Output. We have at most $M$ words of data in fast memory before the segment begins, and may load at most $T$ more during the segment. The number of words we access from the $j$th array during the segment is $p_j|\phi_j(V)|$. Then the number of words we access during this segment is 
\[p_I|\phi_I(V)| + p_F|\phi_F(V)| + p_O|\phi_O(V)| \leq M + T.\]
Let $s_I$, $s_F$, and $s_O$ satisfy $1 \leq s_I + s_F$, $1 \leq s_I + s_O$, $1 \leq s_F + s_O$, and $2 \leq s_I + s_F + s_O$. The discussion in Section \ref{hbl} suggests that we require $s_I + s_F + s_O = 2$ in order to obtain the best asymptotic lower bound. Then by the HBL inequality proved above, we find
\[ |V| \leq |\phi_I(V)|^{s_I}|\phi_F(V)|^{s_F}|\phi_O(V)|^{s_O}.\]

Let $v_j := 2p_j|\phi_j(V)|/(M + T)$. The inequality becomes 
\[|V| \leq \frac{(M + T)^2}{4}(v_I/p_I)^{s_I}(v_F/p_F)^{s_F}(v_O/p_O)^{s_O}.\]
The number of updates $|V|$ possible during this segment is bounded by
$C\max v_I^{s_I}v_F^{s_F}v_O^{s_O}$ subject to the constraint $v_1 + v_2 + v_3 \leq 2$.

We assume $v_I + v_F + v_O = 2$ and apply Lagrange multipliers:
\begin{align*}
 v_I + v_F + v_O &= 2 \\
 s_Iv_I^{s_I - 1}v_F^{s_F}v_O^{s_O} = s_Fv_F^{s_F - 1}v_I^{s_I}v_O^{s_O}= s_Ov_O^{s_O - 1}v_I^{s_I}v_F^{s_F} &= \lambda
%s_Fv_F^{s_F - 1}v_I^{s_I}v_O^{s_O} &= \lambda \\
%s_Ov_O^{s_O - 1}v_I^{s_I}v_F^{s_F} &= \lambda 
\end{align*}
Taking an inner product with $(v_I, v_F, v_O)$ we find:
\[
 (s_I + s_F + s_O)v_I^{s_I}v_F^{s_F}v_O^{s_O} = \lambda(v_I + v_F + v_O)
 \quad \implies \quad 
 v_I^{s_I}v_F^{s_F}v_O^{s_O} = \lambda
\]
since $s_I + s_F + s_O = 2 = v_I + v_F + v_O$. Substituting into each equation and dividing, we find $s_I = v_I$, $s_F = v_F$, $s_O = v_O$.
Then we have shown that the maximum number of updates during the current segment is 
\[|V| \leq \frac{1}{4}(M + T)^2(s_I/p_I)^{s_I}(s_F/p_F)^{s_F}(s_O/p_O)^{s_O}\]

This holds for all triples $(s_j)$ with $s_I + s_F + s_O = 2$ and $s_I, s_F, s_O \leq 1$. In particular, it holds for the triple $(s_j)$ which minimize the right hand side. We apply Lagrange multipliers again ignoring the last three constraints on the $s_j$:
\begin{align*}
 s_I + s_F + s_O &= 2 & \\
 (1 + \log(s_I/p_I))(s_I/p_I)^{s_I}(s_F/p_F)^{s_F}(s_O/p_O)^{s_O} &= \lambda & \\
 (1 + \log(s_F/p_F))(s_I/p_I)^{s_I}(s_F/p_F)^{s_F}(s_O/p_O)^{s_O} &= \lambda & \\
 (1 + \log(s_O/p_O))(s_I/p_I)^{s_I}(s_F/p_F)^{s_F}(s_O/p_O)^{s_O} &= \lambda & 
\end{align*}
Equating and dividing by $(s_I/p_I)^{s_I}(s_F/p_F)^{s_F}(s_O/p_O)^{s_O}$ we find $s_I/p_I = s_F/p_F = s_O/p_O$. This leads to $s_j = 2p_j/(p_I + p_F + p_O)$. Note that these minimizers always satisfy $s_j \leq 1$ for all $j$. Indeed, the triangle condition guarantees $2p_j \leq p_I + p_F + p_O$ for all $j$. Then we have shown that the maximum number of computations during any segment is 
\begin{align*}
 |V| &\leq \frac{1}{4}(M + T)^2(s_I/p_I)^{s_I}(s_F/p_F)^{s_F}(s_O/p_O)^{s_O} \\
 &= \frac{1}{4}(M + T)^2\left(\frac{2}{p_I + p_F + p_O}\right)^{2} = \frac{(M + T)^2}{(p_I + p_F + p_O)^2}
\end{align*}
Since we must do $G$ updates in total, the total number of segments is bounded below:
\[
 L \geq \left\lfloor \frac{G}{|V|} \right\rfloor \geq \frac{(p_I + p_F + p_O)^2G}{(M + T)^2} - 1
\]

Each segment besides the last has $T$ loads/stores, so the total number of words moved is:
\[
 X \geq T\left(\frac{(p_I + p_F + p_O)^2G}{(M + T)^2} - 1\right) = \frac{(p_I + p_F + p_O)^2TG}{(M + T)^2} - T
\]
To choose optimal segment length, we note that $T / (M + T)^2$ is maximized when $T = M$ and we find the following lower bound on the communication cost:
\[
 X \geq \frac{(p_I + p_F + p_O)^2G}{4M} - M\qedhere
\]
\end{proof}

Should the triangle condition fail, we slightly modify the last proof by finding a valid set of minimizers. Note that only one of the three constraints may fail at once: if $p_j > p_k + p_\ell$, then $p_k + p_j > p_\ell$.

\begin{lemma}
\label{lem:singleLFPrecisionNoTriangle}
With the setup in Theorem \ref{thm:singleBoundPrecision}, if $p_j > p_k + p_\ell$ for some distinct $j, k, \ell \in \{I, F, O\}$, the number of words communicated $X$ satisfies \[X \geq \frac{p_j(p_k + p_\ell)G}{M} - M\]
\end{lemma}

\begin{proof}
The proof is the same as the proof of Lemma \ref{lem:singleLFPrecisionTriangle}, except now we prescribe $s_j = 1$ and $s_k + s_\ell = 1$. This guarantees that all conditions for HBL are met. We maximize $(s_k / p_k)^{s_k}(s_\ell/p_\ell)^{s_\ell}$ with respect to $s_k$ and $s_\ell$ as before, and find $s_k/p_k = s_\ell/p_\ell$. This leads to $s_k = p_k/(p_k + p_\ell)$ and $s_\ell = p_\ell / (p_k + p_\ell)$. All constraints are satisfied.
Pick $T = M$.  Then we have shown that the maximum number of computations during any segment is 
\begin{align*}
 |V| &\leq \frac{1}{4}(2M)^2(s_j/p_j)^{s_j}(s_k/p_k)^{s_k}(s_\ell/p_\ell)^{s_\ell} \\
 &= \frac{1}{4}(2M)^2 \frac{1}{p_j} \left(\frac{1}{p_k + p_\ell}\right)^{s_j + s_k} = \frac{(M + T)^2}{4p_j(p_k + p_\ell)}
\end{align*}
We convert this into a communication bound as in Lemma \ref{lem:singleLFPrecisionTriangle}:
\[
 X \geq \frac{p_j(p_k + p_\ell)G}{M} - M\qedhere
\]
\end{proof}

%SMALL FILTER START

%we follow \cite{dd18}
When $M > (C_pG)^{1/2}$, the previous bounds become trivial. When the filter size $w_Fh_F$ is small relative to $M$, we are able to reduce the decay in our bounds from $1/M$ to $1/M^{1/2}$ and extend them to larger memory sizes. To show this third ``small filter" bound, we rewrite the problem and exploit new array access homomorphisms. In particular, we rewrite the loops over $i_6$ and $i_7$ as loops over $q_6, r_6, q_7, r_7$. We have $i_6 = \sigma_w q_6 + r_6$ for $r_6 \in [0, \sigma_w - 1]$ and $q_6 \in [0, w_F/\sigma_w - 1]$, and we similarly divide $i_7$ by $\sigma_h$ for $q_7$ and $r_7$. This has the effect of lifting $Input$ and $Filter$ to higher dimensional arrays, with 6 indices instead of 4.
Under the lift, we make the following accesses to each array during an update:
\begin{align*}
&\mathrm{Input}(i_1, i_2, i_4 + q_6, r_6, i_5 + q_7, r_7) & \\
&\mathrm{Filter}(i_2, i_3, q_6, r_6, q_7, r_7) & \\
&\mathrm{Output}(i_1, i_3, i_4, i_5) &
\end{align*}
In our proof, we will find it valuable to fix the indices $q_6$ and $q_7$.  A new collection of array access homomorphisms ignores these. With an implicit translation by $\vec{q} = (q_6, q_7)$, we define the homomorphisms $\phi_I', \phi_O': \mathbb{Z}^7 \to \mathbb{Z}^4$, $\phi_F' : \mathbb{Z}^7 \to \mathbb{Z}^6$:
\begin{align*}
\phi_I'(i_1, i_2, i_3, i_4, i_5, r_6, r_7) &= (i_1, i_2, i_4, r_6, i_5, r_7) & \\
\phi_F'(i_1, i_2, i_3, i_4, i_5, r_6, r_7) &= (i_2, i_3, r_6, r_7) & \\
\phi_O'(i_1, i_2, i_3, i_4, i_5, r_6, r_7) &= (i_1, i_3, i_4, i_5) &
\end{align*}
Each homomorphism selects a subset of the indices, and every index appears in exactly two of the homomorphisms. This is the case of a tensor contraction analyzed in Section 6.3 of \cite{cdksy13}. They find the optimal HBL exponents to be $s_I = s_F = S_O = 1/2$ and an HBL inequality for finite subsets $V$ of $\mathbb{Z}^7$,
\[ |V| \leq |\phi_I'(V)|^{1/2}|\phi_F'(V)|^{1/2}|\phi_O'(V)|^{1/2}.
\]
We can now begin the proof of the third bound.

\begin{lemma}
 \label{lem:singleSFPrecision}
 The number of words communicated $X$ satisfies \[X \geq \frac{2(p_Ip_Fp_O)^{1/2}(\sigma_w\sigma_h)^{1/2}G }{(w_Fh_FM)^{1/2}} - 2M\]
 \end{lemma}

 \begin{proof}
 We split the 7NL CNN computation into $L$ segments with $T$ loads/stores as before.
 Let $V$ be the set of updates computed during a given segment. For fixed indices $\vec{q} = (q_6, q_7)$, let $V(\vec{q})$ be the slice of $V$ with those two coordinates held constant: $V(\vec{q}) = \pi^{-1}(\vec{q})$ where $\pi$ is the projection of $\mathbb{Z}^9$ onto the $\vec{q}$ coordinates.
 Then to compute every update in $V(\vec{q})$ we must access the entries of Input corresponding to indices $\phi_I'(V(\vec{q}))$, and similarly for Filter and Output (we embed $V(\vec{q})$ in the domain of $\phi_I'$ by ignoring the constant $\vec{q}$ coordinates). We apply our HBL inequality to $V(\vec{q})$,
 \[
   |V(\vec{q})| \leq |\phi_I'(V(\vec{q}))|^{1/2}|\phi_F'(V(\vec{q}))|^{1/2}|\phi_O'(V(\vec{q}))|^{1/2}
 \]
 Note that $V$ is the disjoint union of the $V(\vec{q})$'s. Also, the set of indices of $Filter$ accessed is the disjoint union of the $\phi'_F(V(\vec{q}))$'s. Let $u$ be the number of indices of $Input$ accessed during the segment, $v$ the number of indices of $Output$ accessed during the segment, and $w(\vec{q}) = |\phi_F'(V(\vec{q}))|$ the number of indices of Filter accessed by each slice. We have:
 \begin{align*}
 |\phi'_I(V(\vec{q}))| &\leq u,\quad \abs{\phi'_O(V(\vec{q}))}\leq v \qquad \forall \vec{q} & \\
% |\phi'_O(V(\vec{q}))| &\leq v \qquad \quad \forall \vec{q} & \\
 \Big| {\textstyle\bigcup}_{\vec{q}} V(\vec{q}) \} \Big| &= {\textstyle\sum}_{\vec{q}} w(\vec{q}) & 
 \end{align*} 
 
 We have at most $M$ words in memory before the segment begins, and may load at most $T$ more:
 \[p_I u + p_O v + p_F {\textstyle\sum}_{\vec{q}} w(\vec{q}) \leq M + T\]
 
 Using our HBL inequality,
 \begin{align*}
 |V| &= {\textstyle\sum}_{\vec{q}} \left|V(\vec{q})\right| \leq u^{1/2}v^{1/2}{\textstyle\sum}_{\vec{q}} w(\vec{q})^{1/2} 
 \end{align*}
 The max of $u^{1/2}v^{1/2} \sum_{\vec{q}} w(\vec{q})^{1/2}$ over $p_I u + p_O v + p_F\sum_{\vec{q}} w(\vec{q}) \leq M + T$ bounds the number of updates during our segment. We assume equality and apply Lagrange multipliers:
 \begin{align}
   p_I u + p_O v + p_F{\textstyle\sum}_{\vec{q}} w(\vec{q}) &= M + T & \\
   \frac{1}{2}\frac{v^{1/2}}{u^{1/2}}{\textstyle\sum}_{\vec{q}} w(\vec{q})^{1/2} &= p_I \lambda & \\ 
   \frac{1}{2}\frac{u^{1/2}}{v^{1/2}}{\textstyle\sum}_{\vec{q}} w(\vec{q})^{1/2} &= p_O \lambda & \\
   \frac{1}{2}u^{1/2} v^{1/2} w(\vec{q})^{-1/2} &= p_F \lambda \qquad \forall \vec{q}& 
 \end{align}
 Then by dividing (3) and (4), $p_Iu = p_Ov$.  Equating instances of (5), all the $w(\vec{q}) =: w$ are equal. There are $\frac{w_Fh_F}{\sigma_w\sigma_h}$ pairs of $(\vec{q})$, so equating (3) and (5), \[\frac{1}{p_I}{\textstyle\sum}_{\vec{q}} w(\vec{q})^{1/2} = \frac{1}{p_I}\frac{w_Fh_F}{\sigma_w\sigma_h} w^{1/2} = \frac{u}{p_F w^{1/2}} \] so that $p_F w = \frac{\sigma_w\sigma_h}{w_Fh_F}p_Iu$. Then by (1), the maximizing values are \\
 $u =\frac{M + T}{3p_I}$, $v = \frac{M +T}{3p_O}$, and $w(\vec{q}) = \frac{\sigma_w\sigma_h}{w_Fh_F}\frac{M + T}{3p_F}$ for all $\vec{q}$.
 Using these values, the maximum number of updates during this segment is 
 \begin{align*}
   |V| &\leq u^{1/2}v^{1/2} {\textstyle\sum}_{\vec{q}} w(\vec{q})^{1/2}
   \leq \frac{(M + T)^{3/2}}{3^{3/2}(p_Ip_Fp_O)^{1/2}}\frac{(w_Fh_F)^{1/2}}{(\sigma_w\sigma_h)^{1/2}}
 \end{align*}
 and the number of segments $L$ is bounded below by
 \[
   L \geq \left\lfloor \frac{G}{|V|} \right\rfloor \geq \frac{3^{3/2} (p_Ip_Fp_O)^{1/2} (\sigma_w\sigma_h)^{1/2} G }{(w_Fh_F)^{1/2}(M + T)^{3/2}} - 1
 \]
 Each segment besides the last has at most $T$ loads/stores, so the communication cost is
 \[
   X \geq \frac{3^{3/2} (p_Ip_Fp_O)^{1/2} (\sigma_w\sigma_h)^{1/2} TG }{(w_Fh_F)^{1/2}(M + T)^{3/2}} - T
 \]
 To choose optimal segment length, we note that $T / (M + T)^{3/2}$ is maximized when $T = 2M$ and we find the communication cost
 \[
   X \geq \frac{2(p_Ip_Fp_O)^{1/2}(\sigma_w\sigma_h)^{1/2}G }{(w_Fh_FM)^{1/2}} - 2M\qedhere
 \]
 \end{proof}

Taken together, Lemmas \ref{lem:singleTrivPrecision}, \ref{lem:singleLFPrecisionTriangle}, \ref{lem:singleLFPrecisionNoTriangle}, and \ref{lem:singleSFPrecision} complete the proof of Theorem \ref{thm:singleBoundPrecision}.

\subsection{Attainability}\label{singleattainability}

We now discuss practical algorithms for attaining the previously presented communication bounds. We will focus on four algorithms in particular: im2col, blocking, Winograd convolutions, and FFT convolutions. im2col \cite{juan2020high}, Winograd \cite{meng2019efficient}, and FFT techniques \cite{zjld19} for performing convolutions are all well documented in the literature. We will focus on designing improved blocking algorithms. For loop bounds $(N, c_I, c_O, w_O, h_O, w_F, h_F)$ we call \[B=(b_N, b_{c_I},b_{c_O},b_{w_O},b_{h_O},b_{w_F'},b_{h_F'},b_{w_F''},b_{h_F'})\] a blocking. Note that we are using a small filter trick in the style of \cite{dr16}. To obtain a communication-optimal blocking, we use a linear program. For each bound variable, we have a lower bound of $1$ and an appropriate upper bound, being the corresponding array sizes for most indices, and an expression of the filter size and strides for the filter indices $q_6,q_7,r_6$, and $r_7$. Additionally, we know that the three blocks of the output, image, and filter must all simultaneously fit in memory. We thus have that \begin{align}p_Ob_Nb_{c_O}b_{w_O}b_{h_O}&\leq\frac{p_OM}{p_T}\nonumber\\p_Fb_{c_I}b_{c_O}b_{w_{F'}}b_{w_{F''}}b_{h_{F'}}b_{h_{F''}}&\leq\frac{p_FM}{p_T}\label{opt_prob}\\p_Ib_Nb_{c_I}(b_{w_O}+b_{w_{F'}})(b_{h_O}+b_{h_{F'}})b_{w_{F''}}b_{h_{F''}}&\leq\frac{p_IM}{p_T}\nonumber\end{align} 
We expand the last term into four terms, each bounded by $\frac{M}{12}$. Then, taking logarithms and setting up the linear program, we have for our vector $x=\log B$ elementwise, the problem of maximizing $c^Tx$ where \[c^T=\begin{bmatrix}1&1&1&1&1&1&1&1&1\end{bmatrix}\] subject to the constraints $Ax\leq b$ where \[A=\begin{bmatrix}1&0&1&1&1&0&0&0&0\\
0&1&1&0&0&1&1&1&1\\
1&1&0&1&1&0&1&0&1\\
1&1&0&1&0&0&1&1&1\\
1&1&0&0&1&1&1&0&1\\
1&1&0&0&0&1&1&1&1\end{bmatrix} \textrm{ and } 
%b=\begin{bmatrix}1+\log_M\frac{p_O}{p_T}\\1+\log_M\frac{p_F}{p_T}\\1+\log_M\frac{p_I}{4p_T}\\1+\log_M\frac{p_I}{4p_T}\\1+\log_M\frac{p_I}{4p_T}\\1+\log_M\frac{p_I}{4p_T}\end{bmatrix}\] 
b=\begin{bmatrix}1-\log_Mp_T\\1-\log_Mp_T\\1-\log_M4p_T\\1-\log_M4p_T\\1-\log_M4p_T\\1-\log_M4p_T\end{bmatrix}\] We solve this linear program and take exponentials to find our blocking. Using a linear program, we can asymptotically meet the lower bounds we have derived. To compare the various algorithms for performing the convolution, we symbolically calculate the amount of communication each one requires. We use the FFT communication bound provided in \cite{elango16} and the matrix multiplication communication bound provided in \cite{2019rbpg} to compute the relevant communication volumes. We compute communication volumes using parameters taken from AlexNet. We compare these four with the bounds derived in \ref{serialbounds} and the communication for a naive convolution. The parameters used are taken from \cite{resnet}. The results are presented in Figure 1. 

\begin{figure}
\centering
\includegraphics[width=8.5cm]{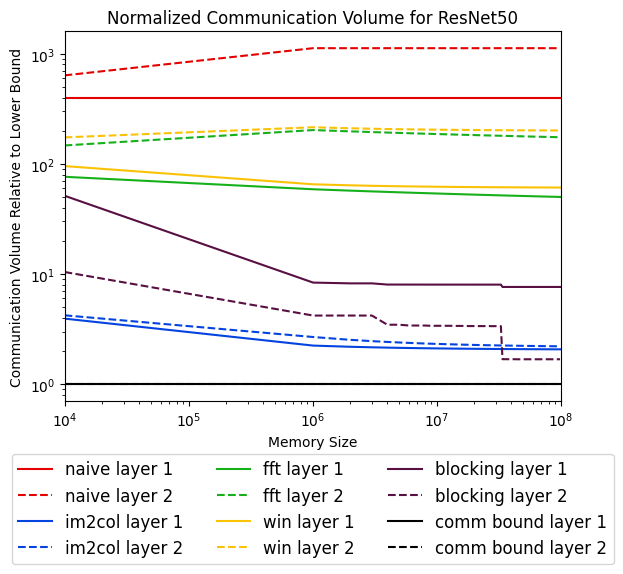}
\caption{Theoretically computed communication volumes for mixed precision ResNet50 layers 1 and 2, relative to the communication bound. We take $\sigma_I=\sigma_F=1$ and $\sigma_O=2$. Layer 1 refers to conv1 and layer 2 refers to conv2\_x as specified in \cite{resnet}. We use a batch size of $1000$. We see that in general, communication volumes are a constant multiple of the communication bound. However, we do see scaling in blocking, and for conv2\_x, the strides of 1 are more favorable to the blocking, and blocking beats im2col for sufficiently large memory sizes. Convolutional layers conv3\_x, conv4\_x, and conv5\_x, which are not depicted, resemble conv2\_x. }
\label{fig:serial1}
\end{figure}

We observe several trends. Blocking and im2col scale better than FFT and Winograd in the memory size, and the relative performance of blocking and im2col is dependent on the ratio $\frac{\sigma_w\sigma_h}{w_Fh_F}$. This is expected because of how the small filter blocking is used. We note that in all cases, the communication bound is not attained precisely. Work remains to either strengthen the communication bound or devise better algorithms to meet the bound.

\section{Parallel Communication Bounds} \label{parallel}

\subsection{Derivation of New Bounds}
\label{parallelbounds}
In this section, we prove Theorems \ref{thm:parallelBoundMDPrecision} and \ref{thm:parallelBoundMIPrecision}. These provide lower bounds for the number of words communicated in a distributed memory parallel processor memory model, with $P$ processors each with local memory size $M$. We assume that all the data is non-overlapping. A single word of communication corresponds to the transmission of one word of data from any one processor to any other. Each array has its own precision.

We now prove Theorem \ref{thm:parallelBoundMDPrecision}. For this result, the input data may begin distributed among the local memories in any configuration and the output data may reside anywhere in memory at the end of the execution. In the case $p_I = p_F = p_O = 1$, Theorem \ref{thm:parallelBoundMDPrecision} becomes
\[
 X \geq \max\left\{\frac{9G}{4PM} - M, \frac{2G(\sigma_w\sigma_h)^{1/2}}{P(w_Fh_FM)^{1/2}} - 2M\right\}.
\]

\begin{proof}[Proof of Theorem \ref{thm:parallelBoundMDPrecision}]
Some processor must perform at least $G / P$ of the updates. Splitting the computations executed by this processor into segments, each having a total of $M$ words communicated into and out of the processor, we need to bound the number of computations possible in a segment. By the same technique as in Lemmas \ref{lem:singleLFPrecisionTriangle} and \ref{lem:singleLFPrecisionNoTriangle}, we find that the number of calculations $|V|$ possible in a segment is $|V| \leq M^2 / C_p$. 
Since we must do at least $G / P$ updates, the total number of segments is bounded below:
\[
 L \geq \frac{G}{P|V|} - 1 \geq \frac{C_pG}{PM^2} - 1
\]
and the number of words communicated by this processor is
\[
 X \geq ML \geq \frac{C_pG}{PM} - M
\]

This is the first term in the desired lower bound. Now, instead splitting into segments with $2M$ communications and using the technique in Lemma \ref{lem:singleSFPrecision}, we find that the number of calculations $|V|$ possible in a segment is
\[
 |V| \leq \frac{M^{3/2}}{(p_Ip_Fp_O)^{1/2}} \left(\frac{w_Fh_F}{\sigma_w\sigma_h}\right)^{1/2}
\]
Since we must do at least $G / P$ updates, and as before the number of words communicated by this processor is
\[
 X \geq 2M\left( \frac{G}{P|V|} - 1 \right) \geq \frac{2(p_Ip_Fp_O)^{1/2}(\sigma_w\sigma_h)^{1/2}G}{P(w_Fh_FM)^{1/2}} - 2M
\]

This is the second term in the desired lower bound.
\end{proof}

When $M > \frac{3G^{1/2}}{2P^{1/2}}$ and $M > \frac{G^{2/3}(\sigma_w\sigma_h)^{1/3}}{P^{2/3}(w_Fh_F)^{1/3}}$, both of the above lower bounds are trivial. This becomes a concern if the memory size per processor or the number of processors is large relative to the size of the computation. Taking inspiration from the methods in \cite{de13} which introduce lower bounds for parallel matrix multiplication corresponding to 2.5D algorithms, we find memory independent lower bounds in Theorem \ref{thm:parallelBoundMIPrecision}. Now we make a load balancing assumption on each of the three arrays: image, filter, and output data are all evenly distributed across the processors. When $p_I = p_F = p_O = 1$, the bound is
\[
 X \geq \max\left\{\frac{G^{1/2}}{P^{1/2}} - \frac{A_p}{P},
 \frac{(G\sigma_w\sigma_h)^{2/3}}{(Pw_Fh_F)^{2/3}} - \frac{A_p}{P}\right\}.
\]

\begin{proof}
 Recall that in Lemma \ref{lem:singleLFPrecisionTriangle} we show that if $V$ is any subset of indices of updates of 7NL CNN, then we can bound the size of $V$:
 \begin{align*} 
   |V| &\leq |\phi_I(V)|^{2/3}|\phi_F(V)|^{2/3}|\phi_O(V)|^{2/3} \\
   &= \frac{1}{(p_Ip_Fp_O)^{2/3}}(p_I|\phi_I(V)|)^{2/3}(p_F|\phi_F(V)|)^{2/3}(p_O|\phi_O(V)|)^{2/3} 
 \end{align*}
 Because the total number of updates is $G$, one processor must do at least $G/P$ updates. Let $V$ be the set of all iterations performed by this processor. Then for at least one $j \in \{I, F, O\}$, we must have
 \[ \frac{G^{1/3}}{P^{1/3}} = |V|^{1/3} \leq \frac{1}{(p_Ip_Fp_O)^{2/9}}(p_j|\phi_j(V)|)^{2/3}. \]
 So the processor must access at least $(p_Ip_Fp_O)^{1/3}G^{1/2}/P^{1/2}$ words from some array during the full computation. Recall that \\
 $A_p := \max\{p_I|I|, p_F|F|, p_O|O|\}$. At most $A_p/P$ words from this array are accessible to the processor at the beginning of the computation by the load balancing assumption, therefore this processor must receive at least $X$ words of data from other processors, where
 \[ X \geq \frac{(p_Ip_Fp_O)^{1/3}G^{1/2}}{P^{1/2}} - \frac{A_p}{P}. \]
 
 Similarly, in Lemma \ref{lem:singleSFPrecision} we show that $V$ has size:
 \begin{align*} 
   |V| &\leq {\textstyle\sum}_{\vec{q}}|\phi_I'(V(\vec{q}))|^{1/2}|\phi_F'(V(\vec{q}))|^{1/2}|\phi_O'(V(\vec{q}))|^{1/2} \\
   &\leq \frac{w_Fh_F}{\sigma_w \sigma_h}|\phi_I(V)|^{1/2}|\phi_F(V)|^{1/2}|\phi_O(V)|^{1/2}
 \end{align*}
 Then for at least one $j \in \{I, F, O\}$, it must be true that
 \[ \frac{G^{1/3}}{P^{1/3}} = |V|^{1/3} \leq \left(\frac{w_Fh_F}{\sigma_w\sigma_h}\right)^{1/3}\frac{1}{(p_Ip_Fp_O)^{1/6}}(p_j|\phi_j(V)|)^{1/2}. \]
 So the processor accesses $(p_Ip_Fp_O)^{1/3}(G\sigma_w\sigma_h)^{2/3}/(Pw_Fh_F)^{2/3}$ words from some array during the full computation. At most $A_p/P$ words from this array are accessible to the processor at the beginning of the computation, so the processor must receive at least $X$ words of data from other processors, where
 \[ X \geq \frac{(p_Ip_Fp_O)^{1/3}(G\sigma_w\sigma_h)^{2/3}}{(Pw_Fh_F)^{2/3}} - \frac{A_p}{P}.\]
 Combining the two lower bounds proves the theorem.
\end{proof}

Note that as in the single processor case, the lower bounds come in pairs, the second eclipsing the first when the filter is sufficiently small. We now discuss the attainability of these results.

\subsection{Attainability}
\label{parallelattainability}

We now discuss algorithms for attaining the previously presented communication bounds, focusing once again on im2col, blocking, Winograd, and FFT. For blocking, instead of blocking in the memory size, we block in the number of processors. For each loop variable, we have a corresponding parallel blocking variable $a_{(\cdot)}$ representing the segment of the loop variable being assigned to each processor. We then have the blocking \[B=(a_N,a_{c_I},a_{c_O},a_{w_O},a_{h_O},a_{w_F},a_{h_F})\] and each processor then does $a_Na_{c_I}a_{c_O}a_{w_O}a_{h_O}a_{w_F}a_{h_F}$ computations. We do not use an additional small filter blocking in this instance. To find the blocking, we once again take logarithms, giving us the following linear program: for a variable $x=\log\{B\}$ we maximize $c^Tx$ for \[c^T=\begin{bmatrix}1&1&1&1&1&1&1\end{bmatrix}\] subject to $Ax\leq b$ where \[A=\begin{bmatrix}-1&0&-1&-1&-1&0&0\\0&-1&-1&0&0&-1&-1\\-1&-1&0&-1&-1&0&0\\-1&-1&0&-1&0&0&-1\\-1&-1&0&0&-1&-1&0\\-1&-1&0&0&0&-1&-1\\-1&-1&-1&-1&-1&-1&-1\end{bmatrix}\] and 
\[b=\begin{bmatrix}1-\log_Pp_T-\log_PNc_Ow_Fh_F
\\1-\log_Pp_T-\log_Pc_Ic_Ow_Oh_O
\\1-\log_P4p_T-\log_PNc_Iw_Fh_F
\\1-\log_P4p_T-\log_PNc_Iw_Fh_O
\\1-\log_P4p_T-\log_PNc_Ih_Cw_F
\\1-\log_P4p_T-\log_PNc_Iw_Oh_O
\\1-\log_PNc_Ic_Ow_Oh_Ow_Fh_F\end{bmatrix}\] 
To compare the four algorithms, we once again symbolically compute the amount of communication each one requires, and we compare it with the bounds given in \ref{parallelbounds}. One should note that the communication models used in the three different bounds for our work, \cite{2019rbpg} for matrix multiplication, and \cite{elango16} for FFT are not quite the same. However, there is a straightforward conversion between them. The difference between the memory models is the assumption as to whether the data initially resides outside of the distributed network or inside of it. To convert between these, we simply add or subtract the total size of the problem $\abs{\text{Image}}+\abs{\text{Filter}}+\abs{\text{Output}}$. The results are presented in Figure 2. 

\begin{figure}
\includegraphics[width=8.5cm]{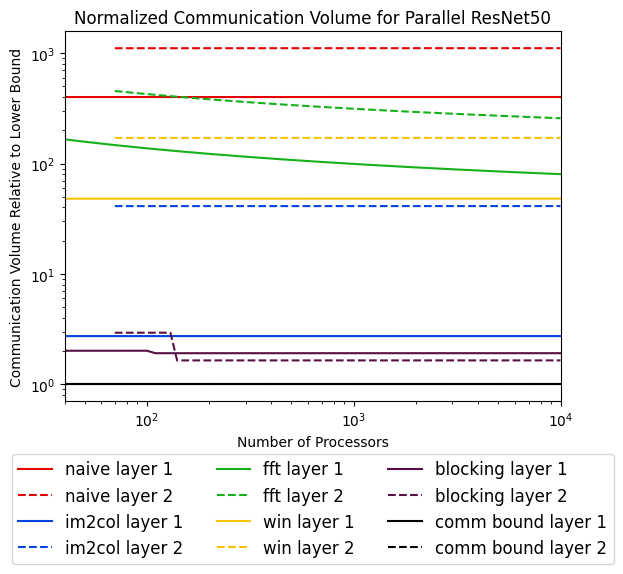}
\caption{Theoretically computed communication volumes for parallel ResNet layers 1 and 2, with $p_I=p_F=1$ and $p_O=2$, as a multiple of the communication bound, as the number of processors increases. Layer 1 refers to conv2 and layer 2 refers to conv2\_x. We use a batch size of $1000$. We see that blocking outperforms im2col considerably, especially for layer 2. Convolutional layers conv3\_x, conv4\_x, and conv5\_x, which are not depicted, resemble conv2\_x. The dashes lines for conv2\_x do not begin at the smallest number of processors because of assumptions on the memory model. We see that blocking performs better than im2col in almost all cases, with significant improvements in conv2\_x when $\sigma_x=\sigma_y=1$ is more favorable to blocking. }
\label{fig:parallel1}
\end{figure}

In both cases, we note that the communication bound goes to 0 very quickly as the number of processors increases. Additionally, we note that for blocking, we have the additional hypothesis that all of the inputs/filter/output elements can reside in the distributed memory, so this method of blocking is not immediately feasible for smaller numbers of processors. However, we see that when blocking is applicable, it rapidly reaches the communication bound as the number of processors increases. We note that Winograd and FFT remain quite far from the communication bound, and that FFT and Winograd have comparable performances, which is validated by the experimental results of \cite{zjld19}, while im2col performs orders of magnitude better. 

\section{Performance Results} \label{performance}
To show real-world applicability of this tiling, we benchmark our results on a GEMMINI \cite{gemmini} machine learning accelerator running on Firesim \cite{firesim}, a cycle-accurate hardware simulator.

GEMMINI's memory architecture consists of two separate memory buffers: a \emph{scratchpad}, which holds the input and the filter, and an \emph{accumulator} buffer, which holds the output at a higher precision (to prevent floating-point rounding issues) and performs additions to it. At each tile, the input and the filter are reloaded from off-chip memory, but the partially summed output is held in the accumulator until it is fully summed (the loop ordering is fixed to ensure that the innermost loop axes in the outer loops correspond to reduction axes), at which point it is rounded and written off-chip in low precision.

We use the default GEMMINI chip configuration, in which, the scratchpad is 256KiB, holding 8-bit words, while the accumulator is 64KiB and holds 32-bit words. However, memory accesses are interleaved with computation using \emph{double-buffering}, in which only half of the scratchpad and the accumulator are accessible to the processor at any one time (with the other half pulling in data from main memory). As a result, for tiling calculations, our memory sizes are halved: the scratchpad can hold 128K words, while the accumulator can hold 8K words.

As a result, we modify the optimization problem \eqref{opt_prob} to account for buffer sharing between the input and the filter and to enforce integral tile sizes. Although this introduces nonlinearity (and an integrality constraint), the built-in numerical optimization routine \texttt{NMaximize} on Mathematica is still able to find a tile in around 400 iterations, or about five seconds on our test laptop.

We compare the performance of the five standard ResNet convolution sizes \cite{resnet} evaluated on GEMMINI using both our tiling and the vendor-supplied tiling system included with GEMMINI. In the vendor implementation, each ResNet convolution size takes roughly the same number of cycles, roughly $500$M for batch size $1000$.

We measure both the estimated communication complexity (the number of scratchpad and accumulator rows allocated by chip's memory controller per tile, multiplied by the total number of operations divided by the size of a tile) and the counted number of clock cycles taken by each computation.
\begin{figure}
\includegraphics[width=8.5cm]{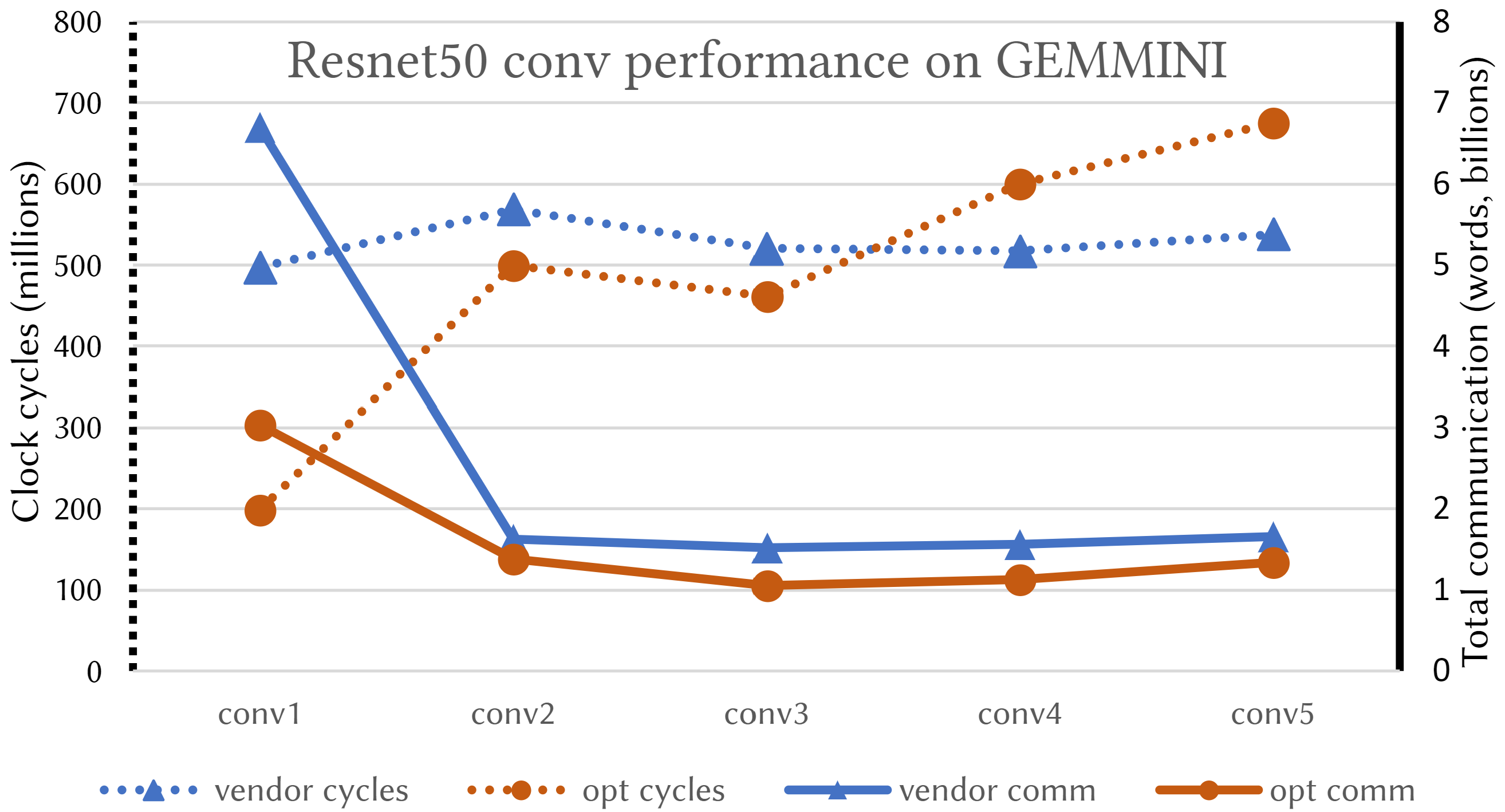}
\caption{Experimentally measured Resnet50 layer performance (both total clock cycles and communication) on GEMMINI accelerator. Our optimization-generated tiling consistently uses less communication than the vendor-supplied tiling, which leads to performance increases for layers where vendor tiling has poor scratchpad utilization (convs 1 through 3).}
\label{fig:gemmini}
\end{figure}
As shown in Figure \ref{fig:gemmini}, our system consistently uses between $45\%$ and $85\%$ as much estimated communication compared to the vendor tiling on all ResNet layers; this can be used as a proxy for energy consumption, which is dominated by communication costs (\cite{YGL+20interstellar} attributes over 80\% of energy costs to communication). Furthermore, for convs 1, 2, and 3 (together comprising roughly half of the workload of a standard ResNet50 instance) where the vendor tiling was unable to take full advantage of the buffer (indicated by low scratchpad utilization per-tile), our tiling reduces clock cycle count (i.e. runtime) by $2.5\times$ for conv1 and $13\%$ for conv2 and conv3. However, for layers 4 and 5, where the vendor tiling already achieves scratchpad utilizations of $99\%$ and $93\%$ respectively, there is little room for improvement; in these cases, our tiling, which does not take into account non-memory related, hardware-specific factors such as optimal microkernel size and memory coalescing, performs worse with respect to clock cycles. In such cases, additional constraints may be added to encode these factors, as in \cite{huang2021cosa}. For instance, for conv5, simply adding a single constraint forbidding the $7\times 7$ image from being tiled (as an entire row will fit in a line of scratchpad) reduces cycles count from $124\%$ to $104\%$ of the vendor figure.

\section{Conclusion} \label{conclusion}

In this work, we have reduced the gap between theoretical communication lower bounds and practical implementations for convolution layers of CNNs. We addressed both a single processor memory model with one cache layer, and a parallel processor distributed memory model. The lower bounds presented in Theorems \ref{thm:singleBoundPrecision}, \ref{thm:parallelBoundMDPrecision}, and \ref{thm:parallelBoundMIPrecision} contain constants and allow for the relative precisions of the data to vary. The single processor lower bound is asymptotically optimal. In Sections \ref{singleattainability} and \ref{parallelattainability} we analyzed the attainability of these lower bounds across popular convolution algorithms such as Im2Col, and found in some cases that a blocking strategy inspired by \cite{cdksy13} communicates less. We included results comparing the performance of our blocking strategy with Im2Col, FFT, and Winograd in Section \ref{performance}.

While our lower bounds are nearly attained, there is more work to be done to close the gap. Future directions of work include determining whether further optimization techniques can tighten the constants in the bounds, and investigating other algorithms to attempt to meet the existing lower bounds. It is also possible that pebbling methods could be used to remove lower order terms from our lower bounds. Finally, it is of interest to extend our results to other memory models, such as single processors with more levels of cache or parallel processors with shared memory.

\begin{acks}
This research was supported by the University of California, Berkeley College of Letters and Sciences Summer Undergraduate Research Fellowship Math Team Grant. We would like to thank Victoria Cheng and the rest of the SURF staff for their valuable support and the donors who made this grant possible. We also thank Hogli Zhao, Jon Hillery, Rahul Jain, and Evangelos Georganas for their insight. Hasan Genc was invaluable for support regarding GEMMINI. 
\end{acks}

\bibliographystyle{ACM-Reference-Format}
\bibliography{reference}

%\appendix

%\section{Proofs of Lower Bounds} \label{proofs}

%\section{Additional Graphs} \label{graphs}

\end{document}